%% file: freqset.tex
\documentclass[letterpaper,11pt]{article}

\usepackage{makeidx}  

\usepackage{amsthm}
\usepackage[margin=1in,dvips]{geometry}

\usepackage{graphicx}
\usepackage{amssymb}
\usepackage{amsmath}
\usepackage{float}
\usepackage{enumerate}
\usepackage{url}
\usepackage{subfig}
\usepackage[hidelinks]{hyperref}

\newtheorem{theorem}{Theorem}[section]
\newtheorem{lemma}[theorem]{Lemma}

\newtheorem{definition}[theorem]{Definition}

\newcommand{\abs}[1]{\left|#1\right|}

\input{libtheorems.tex}

{\makeatletter
 \gdef\xxxmark{%
   \expandafter\ifx\csname @mpargs\endcsname\relax 
     \expandafter\ifx\csname @captype\endcsname\relax 
       \marginpar{xxx}
     \else
       xxx 
     \fi
   \else
     xxx 
   \fi}
 \gdef\xxx{\@ifnextchar[\xxx@lab\xxx@nolab}
 \long\gdef\xxx@lab[#1]#2{{\bf [\xxxmark #2 ---{\sc #1}]}}
 \long\gdef\xxx@nolab#1{{\bf [\xxxmark #1]}}
 \long\gdef\xxx@lab[#1]#2{}\long\gdef\xxx@nolab#1{}%
}

\DeclareMathOperator{\concat}{\|}

\def\R{\mathbb{R}}

\def\eps{\epsilon}

\makeatletter 
\newenvironment{proof*}[1][\proofname]{\par 
  \normalfont \partopsep=\z@skip \topsep=\z@skip 
  \trivlist 
  \item[\hskip\labelsep 
        \itshape 
    #1\@addpunct{.}]\ignorespaces 
}{%
  \endtrivlist\@endpefalse 
} 
\makeatother 

\begin{document}

\title{Optimal Lower Bound for Itemset Frequency Indicator Sketches}

\author{Eric Price\\UT Austin}
\maketitle

\begin{abstract}
  Given a database, a common problem is to find the pairs or
  $k$-tuples of items that frequently co-occur.  One specific problem
  is to create a small space ``sketch'' of the data that records which
  $k$-tuples appear in more than an $\eps$ fraction of rows of the
  database.

  We improve the lower bound of Liberty, Mitzenmacher, and
  Thaler~\cite{LMT14}, showing that $\Omega(\frac{1}{\eps}d \log (\eps
  d))$ bits are necessary even in the case of $k=2$.  This matches the
  sampling upper bound for all $\eps \geq 1/d^{.99}$, and (in the case
  of $k=2$) another trivial upper bound for $\eps = 1/d$.
\end{abstract}

\define{thm:main}{Theorem}{ Any sketch for the
  Itemset-Frequency-Indicator problem must take
  $\Omega(\frac{1}{\eps}d \log (\eps d))$ space for all $1/d \leq \eps \leq
  1/8$, even in the case of $k = 2$.  }

\section{Introduction}

[Check out~\cite{LMT14} for a more complete introduction.]

We are concerned with sketches for itemset frequencies in databases.
The ``itemset frequency'' is the fraction of rows in a database where
a set of items co-occur:

\begin{definition}[Itemset Frequency]
  For a database $\mathcal{D} \in (\{0, 1\}^d)^n$ and an itemset $T
  \subseteq [d]$, the frequency of $T$ in $\mathcal{D}$ is
  \[
  f_T(\mathcal{D}) = \frac{1}{n} \abs{\{i : \forall j \in T, (\mathcal{D}_i)_j = 1 \}}
  \]
\end{definition}

An itemset frequency indicator sketch is a smaller space
representation of $\mathcal{D}$ that lets us identify the itemsets
with large frequency:

\begin{definition}[Itemset-Frequency-Indicator sketches]
  An \emph{Itemset-Frequency-Indicator sketching scheme} is a pair of
  algorithms: one receives $k, \eps$ and a database $\mathcal{D} \in
  (\{0, 1\}^d)^n$ and outputs a sketch $S \in \{0, 1\}^m$, and another
  takes $S$, $\eps$, and a set $T \subset [d]$ with $\abs{T} = k$, and
  returns an estimate of whether $f_T(\mathcal{D}) > \eps$.  In
  particular, it must output YES if
  \[
  f_T(\mathcal{D}) \geq \eps
  \]
  and NO if
  \[
  f_T(\mathcal{D}) \geq \eps/2.
  \]

  For this problem, we require that the first algorithm ``succeed''
  with $3/4$ probability, and if it does then the second algorithm
  should \emph{always} output the correct answer for every query $T$.
\end{definition}

The question is: how large must $m$ to solve this problem?  If we
allowed the queries to fail with a small constant probability, then
per~\cite{LMT14} the space complexity is $\Theta(d/\eps)$.  The goal
of this paper is to get an extra $\log d$ factor from needing to union
bound over $d^k$ queries.

There are two trivial upper bounds, for constant $k$:
\begin{itemize}
\item Sampling takes $O(\frac{1}{\eps} d \log d)$ bits of space.
\item Storing all the answers takes $O(d^k)$ bits of space.
\end{itemize}

We show that $\Omega(\frac{1}{\eps} d \log (\eps d))$ bits are
necessary even in the case of $k = 2$.  This means that sampling is
optimal for all $\eps \geq 1/d^{1 - \alpha}$ for any constant $\alpha
> 0$, while storing all answers is optimal for $\eps \leq 1/d$ and $k
= 2$.

\restate{thm:main}

For $k = 2$, in the relatively minor intermediate regime of $\eps =
1/d^{1 - o(1)}$, it seems likely that neither trivial upper bound is
quite tight.  For $k > 2$, one can probably extend the result to show
that sampling is optimal for $\eps > 1/d^{k-1-\alpha}$; we leave these
questions to future work.

A more interesting open question is for itemset frequency estimation.
If we want to estimate $f_T(\mathcal{D})$ to $\pm \eps$, then sampling
requires $O(\frac{1}{\eps^2} d \log d)$ space but we don't know any
better lower bound than the above $\Omega(\frac{1}{\eps}d \log d)$
bound.  (\cite{LMT14} first showed this for $1/d^{1-\alpha} \ll \eps
\ll 1/\log d$, and our Theorem~\ref{thm:main} removes the upper limit
on $\eps$).

To the best of our knowledge,~\cite{LMT14} contains the only previous
space lower bound for this type of problem.  A number of other aspects
of the problem have been studied, however; see~\cite{LMT14} for an
overview of related work.  Our theorem is a strict improvement over
their Theorem~18, which gets $\Omega(\frac{1}{\eps^{1 - 1/k}}d \log
d)$ for a restricted range of $\eps$.

\section{Notation}
We use $[n]$ to denote $\{1,2,\dotsc,n\}$.  For two vectors $v \in
\R^d$ and $w\in \R^{d'}$, we use $v \concat w$ to denote the $d + d'$
dimensional vector that is the concatenation of $v$ and $v'$.

\section{Proof}

For simplicity of exposition, we begin with the $\eps = \Theta(1)$
case, which was not previously known (\cite{LMT14} required $\eps \ll
1$).  The general $\eps$ case follows a very similar outline.

\begin{lemma}\label{l:consteps}
  Any sketch for the Itemset-Frequency-Indicator problem with $\eps =
  1/8$ must take $\Omega(d \log d)$ space.
\end{lemma}

\begin{proof}
  Let $m = d/2$.  We will encode an arbitrary permutation $\Pi$ of
  $[m]$ into the results of Itemset-Frequency-Indicator.  This forces
  Itemset-Frequency-Indicator to store at least $\log (m!) = \Theta(m
  \log m) = \Theta(d \log d)$ bits.

  For each $i$, define $e_i \in \{0, 1\}^m$ to be the elementary unit
  vector with a $1$ in position $i$.  
  Given a subset $S$ of $[m]$, we associate a vector
  \[
  v_S := (\sum_{i \in S} e_i) \concat (\sum_{i \in \overline{S}} e_{\Pi(i)})
  \]
  where $\concat$ denotes concatenation and $\overline{S} = [m]
  \setminus S$.

  Our database simply consists of $n = \Theta(\log d)$ vectors $v_S$
  for independent, randomly chosen $S$.  In particular, each $S$
  contains each element of $[m]$ with probability $1/2$.

  Now, for each row $v_S$ and any $i, j \in [m]$ consider the
  distribution on the co-occurence of the itemset $\{i, m + j\}$.  If
  $j = \Pi(i)$, this conjunction never appears.  If $j \neq \Pi(i)$,
  on the other hand, then the conjunction appears with $1/4$
  probability.

  After looking at $n = \Theta(\log d)$ such vectors, with high
  probability all itemsets $\{i, m + j\}$ with $j \neq \Pi(i)$ will
  have more than $n/8$ appearances.  Then $f_{\{i, m +
    j\}}(\mathcal{D})$ will be $0$ if $j = \Pi(i)$ and at least $1/8$
  if $j \neq \Pi(i)$.  Therefore an $\eps=1/8$
  Itemset-Frequency-Indicator algorithm will return NO for $\{i, m + j\}$
  precisely when $j = \Pi(i)$, so we can recover $\Pi$ from the
  sketch.  Hence the sketch must have $\Omega(d \log d)$ bits.
\end{proof}

We now extend this approach to general $\eps$ with $1/d \leq \eps \leq
1$.

\state{thm:main}

\begin{proof}[Proof of Theorem~\ref{thm:main}]
  Let $m = \eps d / 2$, which we can assume is an integer by rescaling
  constants.  We will encode $1/\eps^2$ permutations $\Pi_{k, l}$ of
  $[m]$, for $k, l \in [1/\eps]$.  This requires $(1/\eps^2) \log
  ((\eps d/2)!) = \Theta(\frac{1}{\eps}d \log (\eps d))$ bits, giving
  the result.

  Let $e_i \in \{0, 1\}^m$ denote the elementary unit vector with a
  $1$ in position $i$.  For any $S \subset [m]$ and $k \in [1/\eps]$,
  we first define $u^{k, S} \in \{0, 1\}^{d/2}$ by
  \[
  u^{k,S}_i = 1 \text{ if and only if } i = (k-1)m + j \text{ for some }j \in S
  \]
  to represent the set $S$ in ``block'' $k$.  We then define the
  associated vector $v_{k,S} \in \{0, 1\}^d$ by
    \[
  v_{k,S} := u^{k, S}
 \concat (\sum_{i \in \overline{S}} e_{\Pi_{k, 1}(i)})
 \concat (\sum_{i \in \overline{S}} e_{\Pi_{k, 2}(i)})
 \concat \dotsb
 \concat (\sum_{i \in \overline{S}} e_{\Pi_{k, 1/\eps}(i)}).
  \]
  
  We then choose $n = \Theta(\frac{1}{\eps}\log d)$ vectors for the
  database by, for each $k \in [1/\eps]$, choosing $\Theta(\log d)$
  $v_{k,S}$ for uniformly random $S \subseteq [m]$.

  Given the database, to figure out $\Pi_{k, l}(i)$ we query the
  itemset $T_{k, l}(i, j) = \{(k-1)m + i, d/2 + (l-1)m + j\}$ for all
  $j \in [m]$.  We have that $T_{k, l}(i, j)$ appears in $v_{k', S}$
  exactly when $k' = k$ with $i \in S$ and $\Pi_{k,l}^{-1}(j) \notin
  S$.  Thus it never appears if $j = \Pi_{k, l}(i)$, but otherwise it
  appears in each sampled $v_{k, S}$ with probability $1/4$.  Thus
  with high probability, it will appear in at least $\eps n / 8$ of
  the rows.  By a union bound, with high probability $f_{T_{k, l}(i,
    j)}(\mathcal{D}) \geq \eps/8$ for all $i, j, k, l$ with $j \neq
  \Pi_{k, l}(i)$, while it is zero when $j = \Pi_{k,l}(i)$.  Hence an
  $\eps/8$-approximate solution to Itemset-Frequency-Indicator would
  let us recover all the $\Pi_{k, l}$ with high probability,
  retrieving $\Theta(\frac{d}{\eps} \log(\eps d))$ bits of
  information.  Therefore the sketch must store this many bits.
\end{proof}

\bibliography{freqset}
\bibliographystyle{alpha}

\end{document}

%% file: libtheorems.tex
\usepackage{etoolbox}

\makeatletter
\newcommand{\define}[4][ignore]{%
  \ifstrequal{#1}{ignore}{}{
  \@namedef{thmtitle@#2}{#1}}%
  \@namedef{thm@#2}{#4}%
  \@namedef{thmtypen@#2}{lemma}%
  \newtheorem{thmtype@#2}[theorem]{#3}%
  \newtheorem*{thmtypealt@#2}{#3~\ref{#2}}%
}
\newcommand{\state}[1]{%
  \@ifundefined{thmtitle@#1}{
  \begin{thmtype@#1}
    }{
  \begin{thmtype@#1}[\@nameuse{thmtitle@#1}]
  }
    \label{#1}
    \@nameuse{thm@#1}
  \end{thmtype@#1}
}
\newcommand{\restate}[1]{%
  \@ifundefined{thmtitle@#1}{
    \begin{thmtypealt@#1}
    }{
  \begin{thmtypealt@#1}[\@nameuse{thmtitle@#1}]
  }
    \@nameuse{thm@#1}
  \end{thmtypealt@#1}
}
\makeatother